[1994/12/01]
\documentclass[manuscript]{article}
\title{ {\bf Prime Clocks} }
\author{Michael Stephen Fiske}
\date{September 30, 2019}

\usepackage{amsmath,amsthm,amsfonts,amssymb,mathtools,braket}
\usepackage{graphics}

\usepackage[textwidth=6in]{geometry}

\chardef\bslash=`\\ 

\newtheorem{thm}{Theorem}[section]

\newtheorem{lem}[thm]{Lemma}

\theoremstyle{definition}
\newtheorem{defn}{Definition}[section]

\theoremstyle{remark}
\newtheorem{rem}{Remark}[section]

\newtheorem{example}{Example}
\newtheorem{algorithm}{Algorithm}

\newcommand{\eval}[2][\right]{\relax
  \ifx#1\right\relax \left.\fi#2#1\rvert}

\begin{document}

\maketitle




\begin{abstract}
    Physical implementations of digital computers began 
    in the latter half of the 1930's and 
    were first constructed from various forms of logic gates.  
    Based on the prime numbers, we introduce prime clocks and prime clock sums,
    where the clocks utilize time and act as computational primitives instead of gates.   
    The prime clocks generate an infinite abelian group, where for each $n$,
    there is a finite subgroup $S$ such that for each  
    Boolean function $f:\{0, 1\}^n \rightarrow \{0, 1\}$, there exists a finite prime clock sum in 
    $S$ that can represent and compute $f$. A parallelizable algorithm, implemented with a finite prime clock sum, 
    is provided that computes $f$.     
    In contrast, the negation $\neg$, conjunction $\wedge$, and disjunction $\vee$  operations  
    generate a Boolean algebra.  In terms of computation, Boolean circuits 
    computed with logic gates \verb|NOT|, \verb|AND|,  \verb|OR| have a depth.
    This means that a completely parallel computation of Boolean functions 
    is not possible with these gates.  Overall, some new connections between number theory,  
    Boolean functions and computation are established. 
\end{abstract}



\section{Introduction}\label{msf:sect_intro}


\subsection{Notation and Preliminaries}\label{msf:sect_preliminaries}

Symbol $\mathbb{Z}$ denotes the integers and $\mathbb{N}$ the non-negative integers.  For any 
$n \in  \mathbb{N}$ such that $n \ge 2$  and $a \in \mathbb{N}$ such that $0 \le a \le n-1$, 
consider the equivalence class $[a] = \{a + kn:  k \in \mathbb{Z}  \}$ that 
is a subset of $\mathbb{Z}$.  Let $\mathbb{Z}_n = \{[0], [1], \dots,$ $[n-1] \}$.
$a \bmod n$ is the remainder when 
$a$ is divided by $n$.  In the standard manner, $(\mathbb{Z}_n, +_n)$ is an abelian group, 
where binary operator $+_n$
is defined as $[a] +_n [b]$ $= \big{[}$$(a + b) \bmod n$$\big{]}$.   
The brackets are sometimes omitted and 
$[a] \in \mathbb{Z}_n$ is represented with the integer $a$, satisfying $0 \le a \le n-1$. 
The set of all functions  $f: \mathbb{N} \rightarrow \mathbb{Z}_n$ 
is denoted as  ${\mathbb{Z}_n}^{\mathbb{N}}$.  Symbol 
$\overline{c}$ is the constant function $f: \mathbb{N} \rightarrow \mathbb{N}$ where $f(m) = c$ 
for all $m \in \mathbb{N}$.   The set of all $n$-bit strings is $\{0, 1\}^n$.  
It is convenient to identify the 2 bits in $\{0, 1\}$ with the elements $[0]$ and $[1]$ in $\mathbb{Z}_2$.

The least common multiple of positive integers $a$ and $b$ is \verb|lcm|$(a, b)$.  
Let $p_{1} = 2$, $p_{2} = 3$, $p_3 = 5$, $p_4 = 7$, $\dots$ where the  
 $n$th prime number is $p_n$.   Let $p$ be an odd prime.  
$p$ is called a $3 \bmod 4$  prime if $\frac{p-1}{2}$ is odd. 
$p$ is called a $1 \bmod 4$  prime if $\frac{p-1}{2}$ is even.



\subsection{Intuition and Motivation for Prime Clocks}

 Physical implementations of digital computers began in the latter half of the 1930's and 
 early designs were based on various implementations of logic gates 
 \cite{msf:aiken,msf:eniac,msf:atanasoff,msf:turing45,msf:zuse1,msf:zuse2}
 (e.g., mechanical switches, electro-mechanical devices, vacuum tubes).  The transistor 
 was conceptually invented \cite{msf:lilienfeld1,msf:lilienfeld2} in the late 1920's, 
 but the first working prototype \cite{msf:bardeen,msf:riordan} 
 was not demonstrated until 1947.   Transistors act as  
 building blocks for logic gates when they operate above threshold \cite{msf:mead}.    
 The transistor enabled the invention of the integrated 
 circuit \cite{msf:kilby,msf:noyce}, which is the physical basis for modern digital computers.  
 
 As an alternative to gates, prime clocks are based on the prime numbers 
and the notion of a common clock.  Consider the prime number 2 and the clock $[2, 0]$.  
The 2 means that the clock has two states $\{0, 1\}$ 
and the 0 means that the clock starts ticking 
from state 0 at time 0.  Shown in column 2 of table \ref{msf:tab_clocks_2_3_7_13}, the clock $[2, 0]$ ticks
 $0, 1, 0, 1,$ and so on.  
 In column 3 of table \ref{msf:tab_clocks_2_3_7_13}, the clock $[3, 1]$ has 3 states $\{0, 1, 2\}$ and 
 ticks $1, 2, 0, 1, 2, 0$ and so on.

\begin{table}[h]

\centering

\caption{ Some Prime Clocks and Sums in ${\mathbb{Z}_2}^{\mathbb{N}}$     } 

\smallskip 

\label{msf:tab_clocks_2_3_7_13}

\begin{tabular}{  c  c      c  c  c c c c c c c c  c  p{1.0cm}  }
   \hline
   Time   &  $[2,0]$ & $[3, 1]$  & $[2, 0] \oplus [3, 1]$ & $[7, 3]$ & $[13, 6]$ &   $[7, 3] \oplus [13, 6]$  \\
  \hline     
      0  &   0        &   1      &  {\bf 1}   &    3  &   6  &  {\bf 1} \\        
      1  &   1        &   2      &  {\bf 1}   &    4  &   7  &  {\bf 1} \\   
      2  &   0        &   0      &  {\bf 0}   &    5  &   8  &  {\bf 1} \\   
      3  &   1        &   1      &  {\bf 0}   &    6  &   9  &  {\bf 1} \\  
      4  &   0        &   2      &  {\bf 0}   &    0  &  10  &  {\bf 0} \\
      5  &   1        &   0      &  {\bf 1}   &    1  &  11  &  {\bf 0} \\
      $\dots$ \\ 
      
\end{tabular}
\end{table}

Expressed as $\oplus$ in table \ref{msf:tab_clocks_2_3_7_13}, two or more prime clocks can be added
and their sum can be projected into ${\mathbb{Z}_2}^{\mathbb{N}}$.   
The fourth column of table \ref{msf:tab_clocks_2_3_7_13} shows the sum of clocks $[2, 0]$ and $[3, 1]$,  
 projected into ${\mathbb{Z}_2}^{\mathbb{N}}$. This paper primarily focuses on prime clock sums, projected 
into ${\mathbb{Z}_2}^{\mathbb{N}}$, since they can compute Boolean functions.  
These sums have a mathematical property that has a practical application.  This property 
is formally stated in  theorem \ref{msf:theorem_fundamental_prime_clock_sum_boolean}:   
for every natural number $n$, every Boolean function 
$f: \{0, 1\}^n \rightarrow \{0, 1\}$ can be computed with a finite prime clock sum
that lies inside the infinite abelian group 
$( {\mathbb{Z}_2}^{\mathbb{N}}, \oplus)$.  This means prime clocks 
can act as computational primitives instead of gates \cite{msf:shannon,msf:vollmer}. A computer 
can be built from physical devices that implement prime clock sums.  

Prime clock addition $\oplus$ is associative and commutative.  These two group 
properties enable prime clocks to compute in parallel, while gates do not have 
this favorable property.  For example, $\neg(x \wedge y) \neq (\neg x) \wedge y$ because
$\neg(0 \wedge 0) = 1$ while $(\neg 0) \wedge 0 = 0$.  The unary operation $\neg$, 
 conjunction operation $\wedge$, and disjunction operation  $\vee$ 
form a Boolean algebra \cite{msf:halmos}, so circuits built from the \verb|NOT|, 
\verb|AND|, and \verb|OR| gates must have a depth.

\begin{figure}[h]

\centering

\includegraphics[width=0.5\columnwidth]{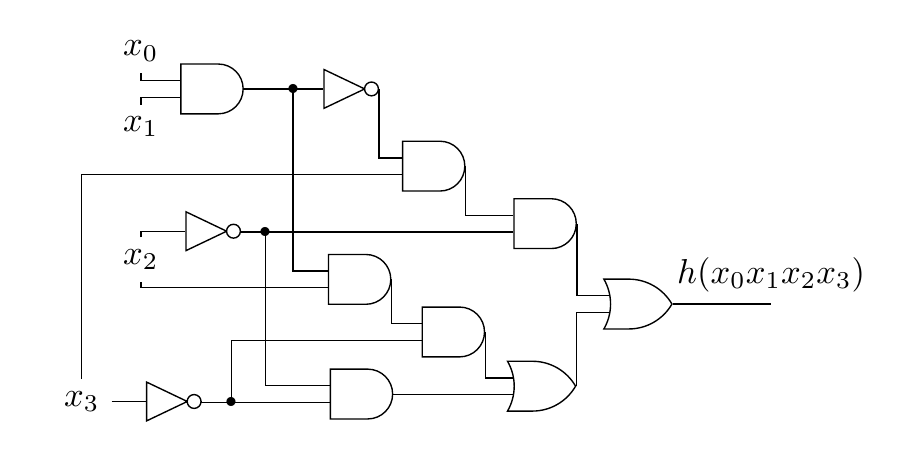}

\caption[ ]{ A gate-based circuit that computes $[7, 3] \oplus [13, 6]$ on $\{0, 1\}^4$. } 

\label{msf:fig_gate_circuit}


\end{figure}

Shown in the last column of table \ref{msf:tab_clocks_2_3_7_13},  the clock sum  $[7, 3] \oplus [13, 6]$, 
helps illustrate the disparity between the {\it parallelization of prime clock sums} 
versus the {\it circuit depth of gates}.   Figure \ref{msf:fig_gate_circuit} shows a gate-based circuit with 
depth 5 that computes $[7, 3] \oplus [13, 6]$ on  $\{0, 1\}^4$. 
This circuit computes Boolean function 
$h: \{0, 1\}^4 \rightarrow \{0, 1\}$, where  $h(x_0$ $x_1$ $x_2$ $x_3) =$  
$\big{[} \big{(} \neg (x_0  \wedge x_1) \big{)} \wedge (\neg x_2) \wedge x_3  \big{]}$
$\vee$
$\big{[} x_0  \wedge x_1 \wedge x_2 \wedge (\neg x_3)  \big{]}$
$\vee$
$\big{[} (\neg x_2) \wedge (\neg x_3)  \big{]}$.              
 Note  $\big{(} [7, 3] \oplus [13, 6] \big{)}(m) = h(x_0$ $x_1$ $x_2$ $x_3)$, 
 whenever $m = x_0 + 2x_1 + 4 x_2 + 8 x_3$.


This disparity enlarges for Boolean functions $f: \{0, 1\}^n  \rightarrow \{0, 1\}$ as $n$ increases.  
Informally, Shannon's theorem \cite{msf:shannon} implies that most 
functions  $f: \{0, 1\}^{n} \rightarrow \{0, 1\}$ require on the order of  
$\frac{2^{n}} {n}$  gates. 
More precisely, let $\beta(\epsilon, n)$ be the 
            number of distinct functions $f: \{0, 1\}^{n} \rightarrow \{0, 1\}$ that can be computed 
            by circuits with at most $(1- \epsilon) \frac{2^n}{n}$ gates built from the 
            \verb|NOT|, \verb|AND|, and \verb|OR| gates.   Shannon's theorem states for any $\epsilon > 0$ 

\begin{equation}
{\underset {n \rightarrow \infty} \lim}  \mbox{\hskip 0.1pc} \frac{ \beta(\epsilon, n)} { 2^{2^{n}} } = 0.
\end{equation} 


Let the gates of a circuit be labeled as $\{g_1, g_2, \dots, g_m \}$ where $m$ is about $\frac{2^n}{n}$. 
The graph connectivity of the circuit specifies that the output of gate $g_1$ connects to the 
input of gate $g_{k_1}$, and so on.  Shannon's theorem implies that for most of these Boolean functions the graph 
connectivity requires an exponential (in $n$) amount of information.  This is readily apparent after comparing 
the number of symbols used in  $[7, 3] \oplus [13, 6]$ versus the symbolic expression  
$\big{[} \big{(} \neg (x_0  \wedge x_1) \big{)} \wedge (\neg x_2) \wedge x_3  \big{]}$
$\vee$
$\big{[} x_0  \wedge x_1 \wedge x_2 \wedge (\neg x_3)  \big{]}$
$\vee$
$\big{[} (\neg x_2) \wedge (\neg x_3)  \big{]}$.

Consider a cryptographic application that uses a function 
$h: \{0, 1\}^{20} \rightarrow \{0, 1\}^{20}$, 
where  $h = (h_0, \dots, h_{19})$ and each 
$h_i: \{0, 1\}^{20} \rightarrow \{0, 1\}$ is highly nonlinear \cite{msf:stanica}. 
Then over 1 million gates can be required to compute $h$, since  
$\frac{2^{20}} {20} = 52428$ and there are 20 distinct $h_i$ functions.  
Using the first 559 prime numbers (i.e., all primes $\le 4051$), finite prime clock sums can 
compute any function $f_{20}: \{0, 1\}^{20}  \rightarrow \{0, 1\}$
even though there are  $2^{2^{20}} = 2^{1048576}$ distinct functions. 
This means a physical realization\footnote{Physical realizations of prime clocks are beyond the scope of this paper.} 
with prime clocks may use the first 599 prime numbers 
to implement an arbitrary $h: \{0, 1\}^{20} \rightarrow \{0, 1\}^{20}$.

Lastly, the structure of our paper is summarized.  Section \ref{msf:sect_prime_clocks} provides formal 
definitions of a prime clock, prime clock sums, and some results about the periodicity of finite 
prime clock sums.  Section \ref{msf:sect_prime_clocks_omega_2} covers prime clock sums projected into 
${\mathbb{Z}_2}^{\mathbb{N}}$, where the main theorem is that any 
Boolean function $f:\{0, 1\}^n \rightarrow \{0, 1\}$ can be computed with a finite prime clock sum.  
Section \ref{msf:sect_prime_clock_sums_algorithm} provides a parallelizable algorithm 
for computing a Boolean function with prime clock sums.



\section{Prime Clocks}\label{msf:sect_prime_clocks}

\begin{defn}\label{msf:defn_prime_clocks}  \hskip 1pc   {\it Prime Clocks}

\noindent Let $p$ be a prime number.   Let $t \in \mathbb{N}$ such that $0 \le t \le p-1$.
Define $[p, t]: \mathbb{N} \rightarrow \mathbb{N}$      
as $[p, t](m) = (m + t) \bmod p$. 
Function $[p, t]$ is called a $p$-{\it clock} that starts ticking with its hand pointing to $t$. 

\end{defn}

Herein the expression {\it prime clock} $[p, t]$ always assumes that $0 \le t \le p-1$.
Thus, if $p \ne q$ or $s \ne t$, then prime clock $[p, s]$ is not equal to $[q, t]$;
equivalently, if $p = q$ and $s = t$, then $[p, s]=[q, t]$.
For the $n$th prime $p_n$,  let $\mathcal{P}_n = \{[p_n, 0], [p_n, 1],$ 
$\dots, [p_n, p_n - 1] \}$ be the distinct $p_n$-clocks.  
The set of all prime clocks is defined as

\begin{equation}
\mathcal{P} = {\underset{n = 1} {\overset{\infty} \cup} } \mathcal{P}_n 
\end{equation}

\noindent  For $n \ge 2$, let $\Omega_n = {\mathbb{Z}_n}^{\mathbb{N}}$.  Define 
$\pi_n: \mathcal{P} \rightarrow \Omega_n$ as the projection of each
$p$-clock into $\Omega_n$ where $\pi_n([p, t](m)) = \big{(} [p, t](m)\big{)} \bmod n$. 
\begin{defn}\label{defn:prime_clock_sum_of_2}  \hskip 1pc

Let $n \in \mathbb{N}$ such that  $n \ge 2$.  
On the set $\mathcal{P}$ of all prime clocks, define the 
binary operator $\oplus_n$ as $\big{(} [p, s] \oplus_n [q, t]\big{)}(m) =$ 
$\big{(} [p, s](m) + [q, t](m) \big{)} \bmod n$, 
where $+$ is computed in $\mathbb{Z}$. Observe 
that $[p, s] \oplus_n [q, t] \in \Omega_n$. 
\end{defn}

\begin{defn}\label{defn:prime_clock_sum}  \hskip 1pc  {\it Finite Prime Clock Sum}

Similarly, with prime clocks $[q_1, t_1]$,  $[q_2, t_2] \dots$ and $[q_l, t_l]$, 
the function $[q_1, t_1]$ $\oplus_n [q_2, t_2] \dots \oplus_n [q_l, t_l]: \mathbb{N} \rightarrow \mathbb{Z}_n$
can be constructed.  For each $m \in \mathbb{N}$, define
$([q_1, t_1] \oplus_n [q_2, t_2] \oplus_n \dots \oplus_n [q_l, t_l])(m)$ 
$= \big{(} [q_1, t_1](m) + [q_2, t_2](m) + \dots + [q_l, t_l](m) \big{)} \bmod n$, 
where $+$ is computed in $\mathbb{Z}$.   
$[q_1, t_1] \oplus_n [q_2, t_2] \dots \oplus_n [q_l, t_l]$ is called 
a {\it finite prime clock sum} in $\Omega_n$. 
\end{defn}

\noindent  
Table \ref{msf:tab_omega_5_clock_table} shows a {\it finite prime clock sum} in $\Omega_5$.

\begin{table}[h]

\centering

\caption{ Some Prime Clocks and their Sum in $\Omega_5$  } 

\smallskip 

\label{msf:tab_omega_5_clock_table} 

 \begin{tabular}{ c  c  c  c   c  c c c c}
   \hline  
Time  & $[5, 3]$ & $[7, 6]$ & $[11, 3]$  & $[13, 0]$  & $[5, 3] \oplus_5 [7, 6] \oplus_5 [11, 3] \oplus_5 [13, 0] $\\
   \hline     
      0    &     3       &   6      &  3        &  0  &  2  \\     
      1    &     4       &   0      &  4        &  1  &  4  \\  
      2    &     0       &   1      &  5        &  2  &  3  \\ 
      3    &     1       &   2      &  6        &  3  &  2  \\ 
      4    &     2       &   3      &  7        &  4  &  1  \\
      5    &     3       &   4      &  8        &  5  &  0  \\
      6    &     4       &   5      &  9        &  6  &  4  \\  
      7    &     0       &   6      &  10       &  7  &  3  \\ 
      8    &     1       &   0      &  0        &  8  &  4  \\ 
      $\dots$     \\  

\end{tabular}

\end{table}

\begin{defn}\label{defn:prime_clock_many_sum} \hskip 1pc
Let  $r_1, \dots r_k$ be $k$ prime numbers and $q_1, \dots q_r$ be $r$ prime numbers.  
Let $f = [r_1, s_1] \oplus_n [r_2, s_2] \dots \oplus_n [r_k, s_k]$. 
Let $g = [q_1, t_1] \oplus_n [q_2, t_2] \dots \oplus_n [q_l, t_l]$.  
Define $f \oplus_n g$ in $\Omega_n$ as 
$\big{(}f \oplus_n g\big{)}$ $(m) = f(m) +_n g(m)$, 
where $+_n$ is the binary operator in the group $(\mathbb{Z}_n, +_n)$.  
\end{defn}

Definition \ref{defn:prime_clock_many_sum} is well-defined with respect to 
definition \ref{defn:prime_clock_sum} 
(i.e., $f \oplus_n g$ $=$  
$[r_1, s_1] \oplus_n [r_2, s_2] \dots \oplus_n [r_k, s_k]$ 
$\oplus_n$
$[q_1, t_1] \oplus_n [q_2, t_2] \dots \oplus_n [q_l, t_l]$
) because   
$(m_1 + m_2) \bmod n$ $=$
$ \big{(} (m_1 \bmod n) + (m_2\bmod n) \big{)} \bmod n$
 for any $m_1, m_2 \in \mathbb{N}$.


\begin{rem}\label{msf:rem_mod_n_addition} 
$(m_1 + m_2) \bmod n$ $=$
$ \big{(} (m_1 \bmod n) + (m_2 \bmod n) \big{)} \bmod n$ \hskip 0.1pc
for any $m_1, m_2 \in \mathbb{N}$.
\end{rem}

\begin{proof}
Euclid's division algorithm implies  $m_1 = k_1 n + r_1$ and  $m_2 = k_2 n + r_2$, 
where $0 \le r_1, r_2 < n$.  Thus,  $(m_1 + m_2) \bmod n$ 
$=$ $\big{(} (k_1 + k_2)n + r_1 + r_2 \big{)} \bmod n$  $=$ $(r_1 + r_2) \bmod n$ 
 $=$ $ \big{(} (m_1 \bmod n) + (m_2 \bmod n) \big{)} \bmod n$   
 \end{proof}

The binary operator $\oplus_n$ can be extended to all of $\Omega_n$.  For any $f, g \in \Omega_n$, define
$ \big{(}f \oplus_n g \big{)}$ $(m) = f(m) +_n g(m)$.  The associative property 
$(f \oplus_n g) \oplus_n h = f \oplus_n (g \oplus_n h)$ follows immediately from the fact that 
$+_n$ is associative.  The zero function $\overline{0}$, where 
$\overline{0}(m) = 0$ in $\mathbb{Z}_n$, is the identity in $\Omega_n$.  
For any $f$ in $\Omega_n$, its unique inverse $f^{-1}$ is defined as $f^{-1}(m) = -f(m)$, 
where $-f(m)$ is the inverse of $f(m)$ in the group
$(\mathbb{Z}_n, +_n)$.  The commutativity of $\oplus_n$ follows from the commutativity of 
$+_n$, so $(\Omega_n, \oplus_n)$ is an abelian group.

Let $\mathcal{Q}$ be a subset of the prime clocks $\mathcal{P}$.  Using the projection $\pi_n$ of $\mathcal{Q}$
into $\Omega_n$, define 
$S_{\mathcal{Q}} = \{ H:  H \supseteq \pi_n(\mathcal{Q})$ \verb|and| $H$ \verb|is a subgroup of| $\Omega_n \}$. 
The subset $\mathcal{Q}$ generates a subgroup of $(\Omega_n, \oplus_n)$.  Namely, 

\begin{equation}
{\underset{H \in S_{\mathcal{Q}} } {\cap} } H
\end{equation}

\noindent  We focus on subgroups of $\Omega_2$, 
generated by a finite number of prime clocks;  consequently, 
the more natural symbol $\oplus$ is used instead of $\oplus_2$.

\begin{defn}\label{defn:periodic} \hskip 1pc  {\it Periodic Functions}

$f \in \Omega_n$ is a {\it periodic function} if there exists a positive integer $b$ such that for 
every $m \in \mathbb{N}$, then $f(m) = f(m + b)$.  Furthermore, if $a$ is the smallest positive integer
such that $f(m) = f(m+a)$ for all $m \in \mathbb{N}$, then $a$ is called the {\it period} of $f$.
After $k$ substitutions of $m+a$ for $m$, this implies for any $m \in \mathbb{N}$ that  
$f(m) = f(m + ka)$ for all positive integers $k$.  
\end{defn}

\begin{table}[h]

\centering

\caption{ Some 2-Clocks, 3-Clocks and Sums in $\Omega_2$   }

\smallskip 

\label{msf:tab_2_3_prime_clocks}

\begin{tabular}{  c  c  c  c  c  c c c c c c c   c  p{1.0cm}  }
   \hline
  Time  &  $[2,0]$ & $[2,1]$ & $[3, 0]$ & $[3, 1]$  & $[2, 0] \oplus [3, 0]$ &  $[2, 1] \oplus [3, 0]$ & $[2, 0] \oplus [3, 1]$ \\
  \hline     
      0  &   0    &   1     &  0  &  1  &  {\bf 0}   &    {\bf 1}  &  {\bf 1}  \\        
      1  &   1    &   0     &  1  &  2  &  {\bf 0}   &    {\bf 1}  &  {\bf 1}  \\   
      2  &   0    &   1     &  2  &  0  &  {\bf 0}   &    {\bf 1}  &  {\bf 0}  \\   
      3  &   1    &   0     &  0  &  1  &  {\bf 1}   &    {\bf 0}  &  {\bf 0}  \\  
      4  &   0    &   1     &  1  &  2  &  {\bf 1}   &    {\bf 0}  &  {\bf 0}  \\
      5  &   1    &   0     &  2  &  0  &  {\bf 1}   &    {\bf 0}  &  {\bf 1}  \\
      \\
      6  &   0    &   1     &  0  &  1  &  {\bf 0}   &    {\bf 1}  &  {\bf 1}  \\
      7  &   1    &   0     &  1  &  2  &  {\bf 0}   &    {\bf 1}  &  {\bf 1}  \\  
      $\dots$ \\ 
\end{tabular}
\end{table}

Table \ref{msf:tab_2_3_prime_clocks} shows that both prime clocks 
$[2,0]$ and $[2,1]$ projected into $\Omega_2$ have period 2. 
Both prime clocks $[3,0]$ and $[3, 1]$ projected into $\Omega_2$ have period 3. 
Each prime clock sum $[2, 0] \oplus [3, 0]$, \hskip 0.3pc  $[2,1] \oplus [3, 0]$ 
\hskip 0.3pc and \hskip 0.3pc  $[2, 0] \oplus [3, 1]$ has period 6.


For any $f \in \Omega_n$, define the relation ${\underset{f}\sim}$ on $\mathbb{N}$ 
such that $x  {\underset{f}\sim} y$ if and only 
if for all $m \in  \mathbb{N}$, $f(m) = f(m + |y - x|)$.  Trivially, ${\underset{f}\sim}$ is reflexive and
symmetric.  

To verify transitivity of ${\underset{f}\sim}$, suppose $x  {\underset{f}\sim} y$ and 
$y  {\underset{f}\sim} z$.  W.L.O.G., suppose $x \le y \le z$.  (The other orderings of $x$, $y$ and $z$
can be handled by permuting $x$, $y$ and $z$ in the following steps.)
This means for all $m \in  \mathbb{N}$, $f(m + y - x) = f(m)$;  and
for all $k \in  \mathbb{N}$,  $f(k) = f(k + z - y)$.  This implies that 
for all $m \in  \mathbb{N}$, $f(m + z - x) = f(m + z - y + y - x) = f(m + y - x) = f(m)$.  
Thus, ${\underset{f}\sim}$ is an equivalence relation.


When $f$ is periodic with period $a$, each equivalence class is 
of the form $[k] =\{k + ma: m \in \mathbb{N} \}$, where $0 \le k < a$.
Thus, $f$ has period $a$ implies there are $a$ distinct equivalence classes
on $\mathbb{N}$ with respect to ${\underset{f}\sim}$.

\begin{rem}\label{msf:rem_periodic_a_divides_b}   \hskip 1pc 
If $a$ is the period of $f$ and $b$ is a positive integer such that 
$f(m) =$ $f(m+b)$ for all $m \in \mathbb{N}$, then $a$ divides $b$. 
\end{rem}

\begin{proof}
First, verify that $a {\underset{f}\sim} b$.  
  By the definition of period, $a \le b$ and for all $m \in \mathbb{N}$, 
  then $f(m + b - a) = f(m + a + b - a) = f(m + b) = f(m)$.  
  From the prior observation, $a$ lies in $[0]$ and $b$ also lies in $[0]$.  
  Thus, $b = ma$ for some positive integer $m$.
\end{proof}

\begin{lem}\label{msf:lem_periodic}  \hskip 1pc 
If $f, g \in \Omega_n$ are periodic, then $f \oplus_n g$ is periodic.  Further, if the period of $f$
is $a$ and the period of $g$ is $b$, then $f \oplus_n g$ has a period that divides \verb|lcm|$(a, b)$. 
\end{lem}

\begin{proof}
Let $a$ be the period of $f$ and $b$ the period of $g$. Let $l_{a,b} =$ \verb|lcm|$(a, b)$.
$l_{a,b} = ia$ and $l_{a,b} = jb$ for positive integers $i, j$.  For any $m \in \mathbb{N}$, 
$\big{(} f \oplus_n g \big{)}(m)$  $=  f(m) +_n g(m)$
$= f(m + ia) +_n g(m + jb)= f(m + l_{a,b}) +_n g(m + l_{a,b})$
$= \big{(} f \oplus_n g \big{)}(m + l_{a,b})$.  
Thus, $f \oplus_n g$ is periodic and remark \ref{msf:rem_periodic_a_divides_b} 
implies its period divides $l_{a,b}$.  
\end{proof}

\smallskip 

\noindent  In regard to lemma \ref{msf:lem_periodic}, if $g = -f$, then the period of $f \oplus_n g$ is 1.  

\begin{rem}\label{msf:rem_num_periodic_functions}   \hskip 1pc 
There are $n^a$ distinct periodic functions $f \in \Omega_n$ whose period divides $a$.
\end{rem}

\begin{proof}
Since $f$ is periodic and its period divides $a$, the values of 
$f(0)$, $f(1)$, $\dots$, $f(a-1)$ uniquely determine $f$.  There are $n$ choices for $f(0)$.
There are $n$ choices for $f(1)$, and so on.  
\end{proof}

\begin{rem}\label{msf:rem_num_periodic_functions}   \hskip 1pc 
Let $p$ be prime. There are $n^p-n$ distinct periodic functions $f \in \Omega_n$ with period $p$.
\end{rem}

\begin{proof}
Consider a finite sequence $c_0$, $c_1$, $\dots$, $c_{p-1}$ of length $p$ where each $c_i \in \mathbb{Z}_n$
This sequence uniquely determines a periodic $f$ such that $f(m + p) = f(m)$ for all $m \in \mathbb{N}$.  In particular,
$f(0) = c_0$, $f(1) = c_1$, $\dots$, $f(p-1)=c_{p-1}$.   There are $n^p$ periodic functions with a period that divides $p$.  
If the period of $f$ is less than $p$, then 
remark \ref{msf:rem_periodic_a_divides_b} implies $f$ has period 1 since $p$ is prime.  There are 
$n$ distinct, constant (period 1) functions in $\Omega_n$  Thus, the remaining $n^p - n$ periodic functions
have period $p$.  
\end{proof}

\begin{rem}\label{msf:rem_periodic}  \hskip 1pc 
The prime clock $[p, t]$, projected into $\Omega_n$, has period $p$.  
\end{rem}


\begin{proof}  
Since $p$ is prime, this follows immediately from remark \ref{msf:rem_periodic_a_divides_b}.  
\end{proof}

\begin{thm}  \hskip 1pc  Finite Prime Clock Sums are Periodic

\smallskip 

\noindent Any finite sum of prime clocks $[q_1, t_1] \oplus_n [q_2, t_2] \oplus_n \dots \oplus_n [q_l, t_l]$ is periodic.  

\end{thm}

\begin{proof}
Use induction and apply remark \ref{msf:rem_periodic} and lemma \ref{msf:lem_periodic}.   
\end{proof}



\section{Prime Clock Sums in $\Omega_2$ }\label{msf:sect_prime_clocks_omega_2}

\begin{rem}\label{msf:rem_identity_Omega_2}  \hskip 1pc  $[p, t] \oplus [p, t] = \overline{0}$ 
for any prime clock $[p, t]$.  


\noindent  Per definition \ref{defn:prime_clock_sum_of_2}, 
$\big{(}[p, k] \oplus [p, k] \big{)} (m)$
$= \big{(} [p, k](m) + [p, k](m) \big{)} \bmod 2 = 0$ in $\mathbb{Z}_2$.  
\end{rem}


\begin{rem}\label{msf:rem_p_p_clocks_sum}  \hskip 1pc  Let $p$ be an odd prime.  
If $p$ is a  $3 \bmod 4$  prime,                  
then $[p, 0] \oplus [p, 1] \oplus \dots \oplus [p, p-1] = \overline{1}$.
If $p$ is a $1 \bmod 4$  prime,    
then $[p, 0] \oplus [p, 1] \oplus \dots \oplus [p, p-1] = \overline{0}$.
\end{rem}

\begin{proof}
$  \big{(}  [p, 0] \oplus [p, 1] \oplus \dots \oplus [p, p-1] \big{)} (0)$ 
       $= (0 + 1 + \dots + p-1) \bmod 2 = \frac{1}{2}(p - 1)p \bmod 2$.  
For each $m > 0$, $\big{(} [p, 0] \oplus [p, 1] \oplus \dots \oplus [p, p-1]\big{)} (m)$ 
is a permutation of the sum  inside $(0 + 1 + \dots + p-1) \bmod 2$.   
\end{proof}



\noindent  For the special case $p = 2$, observe that $[2, 0] \oplus [2, 1] = \overline{1}$.

\begin{defn}\label{defn:non_repeating}  \hskip 1pc 
A finite sum $[q_1, t_1] \oplus [q_2, t_2] \oplus \dots \oplus [q_l, t_l]$ of 
prime clocks is {\it non-repeating} if $i \ne j$ implies $[q_i, t_i]$ is 
not equal to $[q_j, t_j]$. 
\end{defn}

\begin{rem}\label{msf:rem_non_repeating}   \hskip 1pc 
Any finite sum $[q_1, t_1] \oplus [q_2, t_2] \oplus \dots \oplus [q_l, t_l]$ of prime clocks in $\Omega_2$ 
can be reduced to a non-repeating finite sum 
$[q_{i_1}, t_{i_1}] \oplus [q_{i_2}, t_{i_2}] \oplus \dots \oplus [q_{i_r}, t_{i_r}]$, where $r \le l$ such that
for any $m \in \mathbb{N}$, 
$\big{(}[q_1, t_1] \oplus [q_2, t_2] \oplus \dots \oplus [q_l, t_l]\big{)}$ $(m)=$ 
$\big{(}[q_{i_1}, t_{i_1}] \oplus [q_{i_2}, t_{i_2}] \oplus \dots \oplus [q_{i_r}, t_{i_r}]\big{)}$$(m)$.
\end{rem}

\begin{proof}
Since $(\Omega_2, +_2)$ is abelian, if necessary, rearrange the order of 
$[q_1, t_1] \oplus [q_2, t_2] \oplus \dots \oplus [q_l, t_l]$, so that the prime clocks are ordered using
the dictionary order.  If two or more adjacent prime clocks are equal, then the associative property and 
remark \ref{msf:rem_identity_Omega_2} enables the cancellation of even numbers of equal prime clocks. This reduction
can be performed a finite number of times so that the resulting sum is non-repeating.   
\end{proof}


\begin{defn}  \hskip 1pc 
Let $p$ be a prime.  A finite sum of prime clocks 
$[p, t_1] \oplus [p, t_2] \oplus \dots [p, t_{l-1}] \oplus [p, t_{l}]$ is 
called a $p$-{\it clock sum} of length $l$ if for each $1 \le i \le l$, the clock $[p, t_i]$ 
is a $p$-clock and the sum is non-repeating.  
The non-repeating condition implies $l \le p$.
\end{defn}

\begin{lem}\label{msf:lem_p_clock_sums_period_p_or_1}  \hskip 1pc 
Let $p$ be a prime.  A $p$-clock sum with length $p$ has period $1$.  
A $p$-clock sum with length $l$ such that $1 \le l < p$ has period $p$.  
\end{lem}


\begin{proof}
When $p = 2$, the $2$-clock sum $[2, 0]$ has period 2 and the 
$2$-clock sum $[2, 1]$ also has period 2. Recall that $[2, 0] \oplus [2, 1] = \overline{1}$. 
For the remainder of the proof, it is assumed that $p$ is an odd prime.

Let $[p, t_1] \oplus [p, t_2] \oplus \dots [p, t_{l-1}] \oplus [p, t_{l}]$ be a $p$-clock 
sum. When $l = p$, remark \ref{msf:rem_p_p_clocks_sum} implies that 
$[p, t_1] \oplus [p, t_2] \oplus \dots [p, t_{l-1}] \oplus [p, t_{l}]$ has period 1.  
Lemma \ref{msf:lem_periodic} and remark \ref{msf:rem_periodic} imply that 
$[p, t_1] \oplus [p, t_2] \oplus \dots [p, t_{l-1}] \oplus [p, t_{l}]$ has period $p$ or period $1$. 
The rest of this proof shows that $1 \le l \le p-1$ implies that the $p$-clock sum cannot have period $1$.

Thus, it suffices to show that $1 \le l < p$ implies that 
$\big{(} [p, t_1] \oplus [p, t_2] \oplus \dots  \oplus [p, t_{l}] \big{)} (m) \ne$ 
$\big{(} [p, t_1] \oplus [p, t_2] \oplus \dots  \oplus [p, t_{l}] \big{)} (m+1)$ for some $m \in \mathbb{N}$.  
If needed, the $p$-clock sum may be permuted so that 
$[p, s_1]  \oplus [p, s_2] \oplus \dots  \oplus [p, s_{l}] =$
$[p, t_1]  \oplus [p, t_2] \oplus \dots  \oplus [p, t_{l}]$ 
and the $s_i$ are strictly increasingly.  (Strictly increasing means  
$0 \le s_1 < s_2$ $\dots$ $s_{l-1} < s_l \le p-1$.) 

\smallskip 

Case A.  \hskip 0.5pc $l$ is odd.  If $s_l < p-1$, then 
$\big{(} [p, s_1]  \oplus [p, s_2] \oplus \dots  \oplus [p, s_{l}] \big{)}(0)$ 
$= {\overset{l} {\underset{i = 1} \sum} s_i} \bmod 2$  $\ne$ 
${\overset{l} {\underset{i = 1} \sum} (s_i + 1) } \bmod 2$  
$= \big{(} [p, s_1]  \oplus [p, s_2] \oplus \dots  \oplus [p, s_{l}] \big{)}(1)$  because $l$ is odd.  
Otherwise, $s_l = p-1$.  Set $s_0 = 0$.  (The auxiliary index $s_0 = 0$ handles the case
 $s_{k+1} - s_k$ for all $k$ such that $1 \le k < l$.) \hskip 0.3pc 
Set  $m = \max  \big{\{} k \in \mathbb{N}:$  $(s_{k+1} - s_k) \ge 2$ 
     \hskip 0.1pc  \verb|and| \hskip 0.1pc  $0 \le k < l \big{\}}$.
 Since $s_0 = 0$ and $1 \le l < p$, the pigeonhole principle implies $m$ exists.  
Before the $\mod$ $2$ step, the difference between
${\overset{l} {\underset{i=1}{\sum}} } \big{(}(s_i + l - m + 1) \bmod p\big{)}$  
and 
${\overset{l} {\underset{i=1}{\sum}} } \big{(}(s_i + l - m) \bmod p \big{)}$ 
equals $l$.  Hence, 
  $\big{(} [p, s_1]  \oplus [p, s_2] \oplus \dots  \oplus [p, s_{l}] \big{)}(l-m) \ne$ 
$\big{(} [p, s_1]  \oplus [p, s_2] \oplus \dots  \oplus [p, s_{l}] \big{)}(l-m+1)$. 

\smallskip

Case B.  \hskip 0.5pc $l$ is even.  Set $j = (p - 1) - s_l$.  Before the $\bmod$ $2$ step,
the sum  ${\overset{l} {\underset{i=1}{\sum}} } \big{(}(s_i + j) \bmod p \big{)}$ 
differs from the sum  
${\overset{l} {\underset{i=1}{\sum}} } \big{(}(s_i + j + 1) \bmod p \big{)}$  by an odd number.   
Thus,  $\big{(} [p, s_1]  \oplus \dots  \oplus [p, s_{l}] \big{)}(j) \ne$
$\big{(} [p, s_1]  \oplus \dots  \oplus [p, s_{l}] \big{)}(j+1)$.   
\end{proof}



\begin{defn}  \hskip 1pc 
Let $p$ be prime. Suppose the times are strictly increasing:  that is, 
$s_1 < s_2 \dots < s_l$ and $t_1 < t_2 \dots < t_m$.  
 Suppose $\max \{l, m\} \le p$.  Then  
$p$-clock sum $[p, s_1]  \oplus \dots  \oplus [p, s_{l}]$ is {\it distinct} from 
$p$-clock sum $[p, t_1]  \oplus \dots  \oplus [p, t_{m}]$ \hskip 0.3pc  
if $l \ne m$ or for some $i$, $s_i \ne t_i$.  
\end{defn}

7-clock sum $[7,0] \oplus [7,2] \oplus [7, 3]$ is distinct from 
$[7,1] \oplus [7,2] \oplus [7, 3]$.  Table  \ref{msf:tab_7_clocks_in_Omega_2} 
shows that these distinct 7-clock sums are not equal.

\begin{table}[h]

\centering
 
\caption{ Two distinct 7-clock sums that are not equal in $\Omega_2$  }

\smallskip 

\label{msf:tab_7_clocks_in_Omega_2}


 \begin{tabular}{ c c c c c c c c c c c}
   \hline
 Time & $[7, 0]$ & $[7, 1]$ & $[7, 2]$ & $[7, 3]$ & $[7, 0] \oplus [7, 2] \oplus [7, 3]$ & $[7, 1] \oplus [7, 2] \oplus [7, 3]$ \\
\hline     
      0   &  0  &  1   & 0  & 1  &  1  &  0  \\
      1   &  1  &  0   & 1  & 0  &  0  &  1  \\
      2   &  0  &  1   & 0  & 1  &  1  &  0  \\
      3   &  1  &  0   & 1  & 0  &  0  &  1  \\
      4   &  0  &  1   & 0  & 0  &  0  &  1  \\
      5   &  1  &  0   & 0  & 1  &  0  &  1  \\
      6   &  0  &  0   & 1  & 0  &  1  &  1  \\
      $\dots$    \\ 
\end{tabular}
\end{table}


\begin{thm}\label{msf:theorem_distinct_p_sums_are_not_equal}  \hskip 1pc 
For any $3 \bmod 4$ prime $p$, 
if two $p$-clock sums are distinct, then they are not equal in $\Omega_2$. 
The theorem also holds for $p = 2$. 
\end{thm}


\begin{proof}
The special case $p = 2$ can be verified by examining the second and third columns 
of table \ref{msf:tab_2_3_prime_clocks}.    

Let $p$ be a $3 \bmod 4$ prime.  
Assume $p$-clock sum $[p, s_1]  \oplus \dots  \oplus [p, s_{l}]$ is distinct from 
$p$-clock sum $[p, t_1]  \oplus \dots  \oplus [p, t_{m}]$.  By reductio absurdum, suppose

\begin{equation}\label{eqn:p_clock_sums}
[p, s_1]  \oplus \dots  \oplus [p, s_{l}] = [p, t_1]  \oplus \dots  \oplus [p, t_{m}]. 
\end{equation}

\noindent  For each  $s_i \in \{t_1, \dots, t_m \}$, the operation 
$\oplus [p, s_i]$ in $\Omega_2$ can be applied to both sides of equation \ref{eqn:p_clock_sums}.  
Similarly, for each $t_j \in \{s_1, \dots, s_l \}$, the operation
$\oplus [p, t_j]$ can be applied to both sides of equation \ref{eqn:p_clock_sums}.
Since $(\Omega_2, \oplus)$ is an abelian group, equation \ref{eqn:p_clock_sums} 
can be simplified to 
$[p, s_1]  \oplus \dots  \oplus [p, s_{L}]$ $=$ $[p, t_1]  \oplus \dots  \oplus [p, t_{M}]$ such that
$\{s_1, \dots, s_L\} \cap \{t_1, \dots, t_M\} = \emptyset$ and $M + L \le p$.

Set $f = [p, s_1]  \oplus \dots  \oplus [p, s_{L}]$.  
Apply $f \oplus$ to both sides of  
$[p, s_1]  \oplus \dots  \oplus [p, s_{L}]$ $=$ $[p, t_1]  \oplus \dots  \oplus [p, t_{M}]$. 
This simplifies to $f \oplus [p, t_1]  \oplus \dots  \oplus [p, t_{M}] = \overline{0}$.  
Lemma  \ref{msf:lem_p_clock_sums_period_p_or_1} implies that $L + M = p$.  
Since $L + M = p$ and $\{s_1, \dots, s_L\} \cap \{t_1, \dots, t_M\} = \emptyset$ 
and $p$ is a $3 \bmod 4$  prime, 
remark \ref{msf:rem_p_p_clocks_sum} implies that   $f \oplus [p, t_1]  \oplus \dots  \oplus [p, t_{M}] = \overline{1}$. 
This is a contradiction, so $[p, s_1]  \oplus \dots  \oplus [p, s_{l}]$ is not equal to 
 $[p, t_1]  \oplus \dots  \oplus [p, t_{m}]$ in $\Omega_2$.  
\end{proof}


\medskip

Let $\mathcal{S}_l$ be the set of all $p$-clock sums of length $l$, where $1 \le l \le p$.
There are ${p \choose l}$ distinct $p$-clock sums in each set $\mathcal{S}_l$. 
Set $G_p = {\overset{p} {\underset{l=1}\cup}} \mathcal{S}_l$ $\cup$ $\{ \overline{0} \}$. 
For any $f, g \in G_p$,  remark \ref{msf:rem_identity_Omega_2} implies $f \oplus g^{-1}$ in $G_p$.  
Thus, $(G_p, \oplus)$ is an abelian subgroup of $\Omega_2$.
Set $B_p = \{0, 1\}^p$.  For any $a_1 \dots a_p \in B_p$ and
$b_1 \dots b_p \in B_p$, define  $a_1 \dots a_p$ $+_2$ $b_1 \dots b_p$ $=$ 
$c_1 \dots c_p$, where $c_i = (a_i + b_i) \bmod 2$.  $(B_p, +_2)$ is an abelian group 
with $2^p$ elements.   When $p$ is a $3 \bmod 4$ prime, 
define the function $\phi: G_p \rightarrow B_p$ 
where $\phi(\overline{0}) = 0 \dots 0$ $\in B_p$
and $\phi( [p, t_1] \oplus [p, t_2] \oplus \dots [p, t_l] ) = c_1 \dots c_p$ where 
$c_i = \big{(} [p, t_1] \oplus [p, t_2] \oplus \dots [p, t_l] \big{)} (i)$.  We reach  
theorem \ref{msf:theorem_3_mod_4_prime} because $\phi$ is a group isomorphism.

\begin{thm}\label{msf:theorem_3_mod_4_prime}  \hskip 1pc 
Let $p$ be a $3 \bmod 4$ prime.  
The subgroup $G_p$ of $\Omega_2$, generated by the $p$-clocks 
$[p, 0], [p, 1], \dots [p, p-1]$  has order $2^p$ and is isomorphic to $(B_{p}, +_2)$. 
\end{thm}


\begin{proof}
Theorem \ref{msf:theorem_distinct_p_sums_are_not_equal} implies $\phi$ is a group isomorphism. 
\end{proof}


\begin{table}[h]  

\centering
 
\caption{ The 5-clocks projected into $\Omega_2$  } 

\label{msf:tab_5_clocks_in_omega_2}

\smallskip 

 \begin{tabular}{ c c c c c c c c c c c c c c c}
   \hline
 Time & $[5, 0]$ & $[5, 1]$ & $[5, 2]$ & $[5, 3]$ & $[5, 4]$ & $[5, 0] \oplus [5, 1]$ & $[5, 2] \oplus [5, 3] \oplus [5, 4]$ \\
\hline     
      0  &  0  &  1  &  0  &  1  &  0  &  1  & 1  \\
      1  &  1  &  0  &  1  &  0  &  0  &  1  & 1  \\
      2  &  0  &  1  &  0  &  0  &  1  &  1  & 1  \\
      3  &  1  &  0  &  0  &  1  &  0  &  1  & 1  \\
      4  &  0  &  0  &  1  &  0  &  1  &  0  & 0  \\
      5  &  0  &  1  &  0  &  1  &  0  &  1  & 1  \\
      $\dots$ \\ 
\end{tabular}
\end{table}

Theorem \ref{msf:theorem_distinct_p_sums_are_not_equal}  
does not hold when $p$ is a $1 \bmod 4$ prime.  
Table \ref{msf:tab_5_clocks_in_omega_2} shows  $[5, 0] \oplus [5, 1]$ equals
$[5, 2] \oplus [5, 3] \oplus [5, 4]$ in $\Omega_2$.

\begin{thm}\label{msf:theorem_distinct_p_sums_1_mod_4}  \hskip 1pc 
For any $1 \bmod 4$ prime $p$, 
if two $p$-clock sums are distinct and their respective lengths $L$ and $M$ are both $\le \frac{p-1}{2}$, 
then these two $p$-clock sums  are not equal in $\Omega_2$.  
\end{thm}


\begin{proof}
The proof is almost the same as the proof in theorem \ref{msf:theorem_distinct_p_sums_are_not_equal}.  
The conditions $L \le \frac{p-1}{2}$ and $M \le \frac{p-1}{2}$ and the 
reduction $[p, s_1]  \oplus \dots  \oplus [p, s_{L}] \oplus [p, t_1]  \oplus \dots  \oplus [p, t_{M}] = \overline{0}$
leads to an immediate contradiction: $L + M \le p-1$ and  $\{s_1, \dots, s_L\} \cap \{t_1, \dots, t_M\} = \emptyset$ 
means  lemma \ref{msf:lem_p_clock_sums_period_p_or_1}
implies $[p, s_1]  \oplus \dots  \oplus [p, s_{L}] \oplus [p, t_1]  \oplus \dots  \oplus [p, t_{M}]$ has period $p$.   
\end{proof}



\begin{rem}\label{msf:rem_1_mod_4_prime_dual}  \hskip 1pc 
Let $p$ be a $1 \bmod 4$ prime. 
Let $f = [p, s_1]  \oplus \dots  \oplus [p, s_{l}]$ for some $1 \le l \le \frac{1}{2}(p-1)$.
Set $T = \{0, 1, \dots, p-1\} - \{s_1, \dots, s_l \}$.  Now $T = \{t_1, \dots t_m \}$, where
$l + m = p$.   Set $g = [p, t_1]  \oplus \dots  \oplus [p, t_{m}]$.
Then $f = g$ in $\Omega_2$.  
\end{rem}


\begin{proof}
Since $p$ is a $1 \bmod 4$ prime, 
$\big{(}f \oplus g \big{)} (0) = {\overset{p-1} {\underset{0} \sum}} k \bmod 2$ 
$= 0$ in $\mathbb{Z}_2$.
When $k > 1$, the sum of the elements of $f \oplus g$ before projecting into $\Omega_2$ is
a permutation of the elements $\{0, 1, \dots, p-1 \}$.  Hence, for all $k > 1$, 
$\big{(}f \oplus g \big{)} (k) = 0$ in $\mathbb{Z}_2$.  This means $g = f^{-1}$.  
Lastly, $f = f^{-1}$ in $\Omega_2$, so $f = g$ in $\Omega_2$.  
\end{proof}


\smallskip

Let $p$ be a $1 \bmod 4$ prime.  
Set $H_{p-1} = {\overset{ \frac{1}{2} (p-1)} {\underset{l=1}\cup}}\mathcal{S}_l$ 
\hskip 0.1pc $\cup$ \hskip 0.1pc $\{ \overline{0} \}$.
Observe that $|H_{p-1}| = {\overset{ \frac{1}{2} (p-1)} {\underset{l=1}\sum}} {p \choose k}$ 
$+$ $1 =$ $2^{p-1}$.  To verify that $(H_{p-1}, \oplus)$ is a subgroup of $(\Omega_2, \oplus)$, let 
$f, g \in H_{p-1}$.  Since $g = g^{-1}$ in $(\Omega_2, \oplus)$, 
it suffices to show that $f \oplus g$ lies in 
$H_{p-1}$.  If $f$ or $g$ equals $\overline{0}$, closure in $(H_{p-1}, \oplus)$ holds.  
Otherwise, $f = [p, s_1]  \oplus \dots  \oplus [p, s_{l}]$ for some $1 \le l \le \frac{1}{2}(p-1)$
and $g = [p, t_1]  \oplus \dots  \oplus [p, t_{m}]$ for some $1 \le m \le \frac{1}{2}(p-1)$.
As mentioned before, the sum $f \oplus g$ may be reduced to 
$[p, s_1]  \oplus \dots  \oplus [p, s_{L}]$ $\oplus$ $[p, t_1]  \oplus \dots  \oplus [p, t_{M}]$,
where $\{s_1, \dots, s_L\} \cap \{t_1, \dots, t_M\} = \emptyset$ and $L + M \le p$.
If $L + M \le \frac{1}{2} (p-1)$, closure in $(H_{p-1}, \oplus)$ holds.  
Otherwise, if $L + M > \frac{1}{2} (p-1)$, remark \ref{msf:rem_1_mod_4_prime_dual} implies that 
there is a $p$-clock sum $h = f \oplus g$, where $h$'s length is $p - (L+M)$ 
and $p - (L+M) \le  \frac{1}{2} (p-1)$.

\smallskip

Similar to the group isomorphism $\phi$, define $\psi: H_{p-1} \rightarrow B_{p-1}$ such that 
$\psi(\overline{0}) = 0 \dots 0$ $\in B_p$.  For each $p$-clock sum in $\mathcal{S}_l$, 
where $1 \le l \le \frac{1}{2} (p-1)$, define  
$\psi( [p, t_1] \oplus [p, t_2] \oplus \dots [p, t_l] ) = c_1 \dots c_{p-1}$ where  
$c_i = \big{(} [p, t_1] \oplus [p, t_2] \oplus \dots [p, t_l] \big{)} (i)$. %
It is straightforward to verify that $\psi$ is a group isomorphism onto $B_{p-1}$. 
The group isomorphism $\psi: H_{p-1} \rightarrow B_{p-1}$ leads to the following theorem.

\begin{thm}\label{msf:theorem_1_mod_4_prime}  \hskip 1pc 
Let $p$ be a $1 \bmod 4$ prime.  The subgroup 
$H_{p-1}$ of $\Omega_2$, generated by the $p$-clocks 
$[p, 0], [p, 1], \dots [p, p-1]$  
has order $2^{p-1}$ and is isomorphic to $(B_{p-1}, +_2)$. 
\end{thm}


\begin{thm}\label{msf:theorem_fundamental_prime_clock_sum_boolean}
For positive integer $n$ and any function $f: \{0, 1\}^n$ $\rightarrow \{0, 1\}$, 
there exists a finite sum of prime clocks in $\Omega_2$ that can compute $f$.
\end{thm}

\begin{proof}
Euclid's second theorem implies there is a prime $p > 2^n$, where    
$p$ is a $3 \bmod 4$ or $1 \bmod 4$ prime.  Hence,   
theorem \ref{msf:theorem_3_mod_4_prime} or \ref{msf:theorem_1_mod_4_prime} completes the proof.     
\end{proof}

\smallskip

\noindent  Furthermore, finding a finite prime clock sum 
that computes $f$ is Turing computable       
and there are efficient Turing computable algorithms that 
can decide whether a natural number $n$ is prime \cite{msf:saxena}.


\section{ Prime Clock Sums Compute Boolean Functions in $\Omega_2$ }\label{msf:sect_prime_clock_sums_algorithm}

Let $F_n$ denote the set of all Boolean functions in $n$ variables.  
Formally, the set $F_n =$ $\big{ \{ } f$ $|$  $f :\{0, 1\}^n \rightarrow \{0, 1\}$ $\big{ \} }$ and 
$F_n$ contains  $2^{2^n}$ distinct functions.  For prime clock sums, 
it is convenient to think of $f \in F_n$ as a 
binary string of length $2^n$, called the {\it truth-table} of $f$.   
Table \ref{msf:tab_binary_booleans} 
shows all 16 Boolean functions in $F_2$, their truth tables and 
corresponding prime clock sums that compute each function.

\begin{table}[h]  
 
\centering

 \caption{ \hskip 0.5pc  $f_k: \{0, 1\}^2    \rightarrow \{0, 1\}$.    }   
 \label{msf:tab_binary_booleans}
 
\smallskip 


\begin{tabular}{p{6.5cm}c l p{0.5cm} l p{3.0cm} }
 \hline
 Boolean Function & & Truth Table &  & Prime Clock Sum \\
  \hline
  $f_1(x, y) = 1$          & &   $1111$    & &  $[2, 0] \oplus [2, 1]$  \\
  $f_2(x, y) =  0$         & &   $0000$    & &  $[2, 0] \oplus [2, 0]$  \\
   $f_3(x, y) =  x$        & &   $0011$    & &   $[2, 1] \oplus [3, 1]$       \\
   $f_4(x, y) =  y$        & &   $0101$    & &   $[2, 0]$       \\
   $f_5(x, y) = \neg x$    & &   $1100$    & &   $[2, 0] \oplus [3, 1]$       \\
    $f_6(x, y) = \neg y$   & &   $1010$    & &   $[2, 1]$   \\    
    $f_7(x, y) =  x \wedge y$    & &  $0001$  & &  $[2, 0] \oplus [3, 0]$   \\
    $f_8(x, y) =  x \vee y$      & &   $0111$     & &   $[2, 0] \oplus [3, 2]$   \\
    $f_9(x, y) = \neg x \vee y$  & &   $1101$      & &   $[3, 0] \oplus [3, 1]$     \\
    $f_{10}(x, y) =  x \vee \neg y$ & &  $1011$   & &  $[3, 1] \oplus [3, 2]$     \\
    $f_{11}(x, y) =  (x \wedge y) \vee \neg (x \vee y)$ & &  $1001$  & &  $[3, 1]$      \\
    $f_{12}(x, y) =  (x \vee y) \wedge \neg (x \wedge y)$  & &   $0110$  & &  $[3, 0] \oplus [3, 2]$     \\
    $f_{13}(x, y) =  \neg (x \vee y)$   & &   $1000$   & &  $[2, 1] \oplus [3, 2]$  \\
    $f_{14}(x, y) = \neg (x \wedge y)$  & &   $1110$   & &  $[2, 1] \oplus [3, 0]$   \\
    $f_{15}(x, y) = \neg x \wedge y$    & &   $0100$   & &  $[3, 0]$    \\     
    $f_{16}(x, y) = x \wedge \neg y$    & &   $0010$   & &  $[3, 2]$    \\    
   \end{tabular}  

   \bigskip 

   The truth table for $\{0, 1\}^2$ is ordered as $\{00, 01, 10, 11\}$. 
  
      

\end{table}

Consider $[p, s] \oplus [q, t]$ in $\Omega_2$.  
The {\it first} $2^n$ {\it elements}  of  $[p, s] \oplus [q, t]$ refer to the bit 
string $\big{(} [p, s] \oplus [q, t] \big{)}(0)$, 
\hskip 0.3pc $\big{(} [p, s] \oplus [q, t] \big{)} (1)$, \hskip 0.3pc $\dots$, 
\hskip 0.3pc $\big{(} [p, s] \oplus [q, t] \big{)} (2^n - 1)$ of length $2^n$. 
The first $2^n$ elements of  $[p, s] \oplus [q, t]$ 
represent a Boolean function $f \in F_n$. 
In the general case, if $q_1, \dots, q_L$ are primes, the first 
$2^n$ elements of $[q_1, t_1] \oplus [q_2, t_2] \oplus \dots \oplus [q_L, t_L]$
also represent a Boolean function $f_n \in F_n$.  Consider the first $2^n$ elements of prime clock sum  
$[q_1, t_1] \oplus [q_2, t_2] \oplus \dots \oplus [q_L, t_L]$.   
Algorithm \ref{msf:alg_prime_clock_sum_boolean} computes the $i$th element of this truth table in $F_n$.

\begin{algorithm}\label{msf:alg_prime_clock_sum_boolean}  \hskip 1pc A Prime Clock Sum in $\Omega_2$ Computes a Boolean Function 

\smallskip 

\noindent  {\bf INPUT:} \hskip 1pc $i$

\smallskip 

\noindent  \verb|  set| \hskip 0.3pc $r_1 = (t_1 + i) \bmod q_1$   

\noindent  \verb|  set| \hskip 0.3pc $r_2 = (t_2 + i) \bmod q_2$ 

\noindent  \verb|  . . .|

\noindent  \verb|  set| \hskip 0.3pc $r_L = (t_L + i) \bmod q_L$ 

\noindent  \verb|  set | $y = (r_1 + r_2 + \dots + r_L) \bmod 2$

\medskip 

\noindent  {\bf OUTPUT:} \hskip 1pc $y$

\smallskip 


\end{algorithm}


\begin{example}
We demonstrate 2-bit multiplication with prime clock sums, computed with algorithm 
\ref{msf:alg_prime_clock_sum_boolean}. 
In table \ref{msf:tab_2_bit_multiply}, for each 
$u \in \{0, 1\}^2$ and each $l \in \{0, 1\}^2$,  the product  $u * l$ is shown in each row, 
whose 4 columns are labelled by $\mathcal{M}_3$, $\mathcal{M}_2$, $\mathcal{M}_1$ and $\mathcal{M}_0$.  
With input $i$ of 4 bits (i.e., $u$ concatenated with $l$), the output of the 2-bit 
multiplication is a 4-bit string 
$\mathcal{M}_3(i)$ $\mathcal{M}_2(i)$ $\mathcal{M}_1(i)$  $\mathcal{M}_0(i)$, 
shown in each row of table \ref{msf:tab_2_bit_multiply}.

One can verify that, according to algorithm \ref{msf:alg_prime_clock_sum_boolean}, 
prime clock sum $[2, 0] \oplus [7, 3] \oplus  [7, 4] \oplus  [7, 5] \oplus [11, 10]$   
computes function $\mathcal{M}_0: \{0, 1\}^2 \times \{0, 1\}^2 \rightarrow \{0, 1\}$.   
Similarly, $[2, 0] \oplus  [2, 1] \oplus [3, 0] \oplus [5, 2] \oplus  [11, 0] \oplus [11, 1]$ 
computes function $\mathcal{M}_1$.
Prime clock sum $[5, 0]  \oplus [7, 0] \oplus  [7, 2]  \oplus [11, 4]$ computes function $\mathcal{M}_2$.  
Lastly, $[2, 1] \oplus [5, 0]  \oplus  [11, 1] \oplus [11, 6]$ computes function $\mathcal{M}_3$.

\end{example}

\begin{table}[h]

\centering

\caption{ 2-Bit Multiplication.     }

Multiplication functions $\mathcal{M}_i: \{0, 1\}^4 \rightarrow \{0, 1\}$.  

\medskip

\label{msf:tab_2_bit_multiply}

\begin{tabular}{  l  c   c  c  c  c  c  c  c  c  c  c   c  c  }
 \hline
 $u$ & $l$  &  \verb|  | & $\mathcal{M}_3$ & \verb| | & $\mathcal{M}_2$ & \verb| | & $\mathcal{M}_1$ & \verb| | & $\mathcal{M}_0$  \\
 \hline     
00  & 00  & & 0 & & 0 & & 0 & & 0 \\ 
00  & 01  & & 0 & & 0 & & 0 & & 0 \\ 
00  & 10  & & 0 & & 0 & & 0 & & 0 \\ 
00  & 11  & & 0 & & 0 & & 0 & & 0 \\ 
01  & 00  & & 0 & & 0 & & 0 & & 0 \\ 
01  & 01  & & 0 & & 0 & & 0 & & 1 \\ 
01  & 10  & & 0 & & 0 & & 1 & & 0 \\ 
01  & 11  & & 0 & & 0 & & 1 & & 1 \\ 
 \\ 
10  & 00  & & 0 & & 0 & & 0 & & 0 \\ 
10  & 01  & & 0 & & 0 & & 1 & & 0 \\ 
10  & 10  & & 0 & & 1 & & 0 & & 0 \\ 
10  & 11  & & 0 & & 1 & & 1 & & 0 \\ 
11  & 00  & & 0 & & 0 & & 0 & & 0 \\ 
11  & 01  & & 0 & & 0 & & 1 & & 1 \\ 
11  & 10  & & 0 & & 1 & & 1 & & 0 \\ 
11  & 11  & & 1 & & 0 & & 0 & & 1 \\ 
\end{tabular}

\end{table}


 The $i$th element of 
$\big{(}[q_1, t_1] \oplus [q_2, t_2] \oplus \dots \oplus [q_L, t_L]\big{)}$'s 
truth table is stored in the variable $y$ when algorithm \ref{msf:alg_prime_clock_sum_boolean} halts.  
Algorithm \ref{msf:alg_prime_clock_sum_boolean} is presented in a serial form.  Nevertheless, the computation 
of the $L$ instructions \verb| set| \hskip 0.3pc $r_k =  (t_k + i) \bmod q_k$, where $1 \le k \le L$,    
can be computed in parallel when there is a separate physical device for each of these $L$ prime 
clocks $[q_1, t_1]$, $[q_2, t_2]$  $\dots$ $[q_L, t_L]$.  
Subsequently, the parity of $y$ can be determined in a second computational 
step that executes a parallel add of  $r_1 + r_2 + \dots + r_L$, followed by 
setting $y$ to the least significant bit of the sum   $r_1 + r_2 + \dots + r_L$.  

\smallskip  

As an alternative implementation of algorithm \ref{msf:alg_prime_clock_sum_boolean},
when there is a more suitable physical device for prime clocks, the $k$th clock can compute the $k$th 
bit $b_k =$ $\big{(} (t_k + i) \bmod q_k \big{)} \bmod 2$ and then a parallel exclusive-or \cite{msf:reif} 
can be applied to the $L$ bits $b_1$, $b_2$, $\dots$, $b_L$.  
In contrast, a gate-based Boolean circuit requires at least $d$ computational steps where $d$ is 
the depth of the circuit.  



\bibliographystyle{abbrv}




\end{document}